\documentclass[11pt,letterpaper]{article}

\usepackage[margin=1in]{geometry}
\usepackage{lmodern}
\usepackage[T1]{fontenc}
\usepackage[utf8]{inputenc}
\usepackage[tbtags]{amsmath}
\allowdisplaybreaks
\usepackage{amssymb,amsthm}
\usepackage{hyperref}
\usepackage[svgnames]{xcolor}
\usepackage[capitalise,nameinlink]{cleveref}
\definecolor{links}{RGB}{11, 85, 255}
\definecolor{cites}{RGB}{0, 200, 0}
\definecolor{urls}{RGB}{255, 116, 0}
\hypersetup{colorlinks={true},linkcolor={links},citecolor=[named]{cites},urlcolor={urls}}
\usepackage[numbers,compress]{natbib}
\usepackage{doi}
\usepackage{tikz} 
\usepackage{pgfplots}
\pgfplotsset{compat=1.14}
\usepgfplotslibrary{fillbetween}
\usepackage{hyphenat}
\usepackage[font=footnotesize]{caption}
\usepackage{subcaption}
\usepackage{appendix}
\usepackage[ruled,vlined,linesnumbered]{algorithm2e}
\providecommand{\DontPrintSemicolon}{\dontprintsemicolon}
\SetKwFor{WP}{With probability}{}{:}{}
\SetKwBlock{Prob}{With probability}{end}
\makeatletter
\newcommand*{\algotitle}[2]{%
	\stepcounter{algocf}%
	\hypertarget{algocf.title.\theHalgocf}{}%
	\NR@gettitle{#1}%
	\label{#2}%
	\addtocounter{algocf}{-1}%
}

\newtheorem{theorem}{Theorem}
\newtheorem{lemma}{Lemma}

\theoremstyle{definition}
\newtheorem{definition}{Definition}

\newtheorem{remark}{Remark}

\newcommand{\mysetminusD}{\hbox{\tikz{\draw[line width=0.6pt,line cap=round] 
(3pt,0) -- (0,6pt);}}}
\newcommand{\mysetminusT}{\mysetminusD}
\newcommand{\mysetminusS}{\hbox{\tikz{\draw[line width=0.45pt,line 
cap=round] (2pt,0) -- (0,4pt);}}}
\newcommand{\mysetminusSS}{\hbox{\tikz{\draw[line width=0.4pt,line 
cap=round] (1.5pt,0) -- (0,3pt);}}}
\newcommand*\samethanks[1][\value{footnote}]{\footnotemark[#1]}
\newcommand{\mysetminus}{\mathbin{\mathchoice{\mysetminusD}{\mysetminusT}{\mysetminusS}{\mysetminusSS}}}
\DeclareMathOperator*{\argmax}{arg\,max}

\newcommand{\E}[1]{\mathbb E \left[ #1 \right]}
\newcommand{\Ei}{\mathcal{E}_i}
\newcommand{\Ea}{\mathcal{E}_a}
\newcommand{\p}[1]{\mathbb P \left( #1 \right)}

\author{
	Georgios Amanatidis\thanks{Department of Mathematical Sciences, University of Essex, UK, and Institute for Logic, Language and Computation, University of Amsterdam, The Netherlands. Email:   
\texttt{\href{mailto:georgios.amanatidis@essex.ac.uk}{georgios.amanatidis}@essex.ac.uk}}
	\and
	Federico Fusco\thanks{Department of Computer, Control, and Management 
Engineering ``Antonio Ruberti'', Sapienza University of Rome, 
Italy. Email:   
\texttt{\{\href{mailto:fuscof@diag.uniroma1.it}{fuscof},
        \href{mailto:lazos@diag.uniroma1.it}{lazos},
		\href{mailto:leonardi@diag.uniroma1.it}{leonardi},
		\href{mailto:rebeccar@diag.uniroma1.it}{rebeccar}\}@diag.uni\allowbreak 
		roma1.it}
}
\and
Philip Lazos\samethanks
\and
Stefano Leonardi\samethanks
\and
Rebecca Reiffenh\"auser\samethanks
}

\title{Fast Adaptive Non-Monotone Submodular Maximization \\ Subject to a Knapsack Constraint\thanks{
This work was supported by the ERC Advanced 
Grant 788893 AMDROMA ``Algorithmic and Mechanism Design Research in 
Online Markets'' and the MIUR PRIN project ALGADIMAR ``Algorithms, Games, 
and Digital Markets.''}}

\date{\today}

\begin{document}

\maketitle
\begin{abstract}
Constrained submodular maximization problems encompass a wide variety of 
applications, including personalized recommendation, team formation, and 
revenue maximization via viral marketing. The massive instances occurring in 
modern day applications can render existing algorithms prohibitively slow, while 
frequently, those instances are also inherently stochastic. Focusing on these 
challenges, we revisit the classic problem of maximizing a (possibly 
non-monotone) submodular function subject to a knapsack constraint. We 
present a simple randomized greedy algorithm that achieves a $5.83$ 
approximation and runs in $O(n \log n)$ time, i.e., at least a factor $n$ faster than 
other state-of-the-art algorithms. The robustness of our approach allows us to 
further transfer it  to a stochastic version of the problem. There, we obtain a 
9-approximation to the best adaptive policy, which is the first constant 
approximation for non-monotone objectives. Experimental evaluation of our 
algorithms showcases their improved performance on real and synthetic data.
\end{abstract}

\section{Introduction}
Constrained submodular maximization is a fundamental problem at the heart of 
discrete optimization.
The reason for this is as simple as it is clear: submodular functions capture the 
notion of \emph{diminishing returns} present in a wide variety of real-world 
settings.

Consequently to its striking importance and coinciding NP-hardness 
\citep{Feige98}, extensive research has been conducted on submodular 
maximization since the seventies (e.g., \citep{Edmonds71,NemhauserWF78}), 
with focus lately shifting towards handling the massive datasets emerging in 
modern applications. With a wide variety of possible constraints, often regarding 
cardinality, independence in a matroid, or knapsack-type restrictions, the 
number of applications is vast. To name just a few, there are recent works on 
feature selection in machine learning \citep{DK08, DK18, 
Khanna2017ScalableGF}, influence  maximization in viral marketing 
\citep{BabaeiMJS13,KempeKT15}, and data summarization 
\citep{SiposSSJ12,MirzasoleimanKSK13,TschiatschekIWB14}. Many of these 
applications have \emph{non-monotone} submodular objectives, meaning that 
adding an element to an existing set might actually decrease its value. Two such 
examples are discussed in detail in \cref{sec:experiments}.

Modern-day applications increasingly force us to face two distinct, but often 
entangled challenges. 
First, the massive size of occurring instances fuels a need for very fast algorithms. As the running time is dominated by the \emph{objective 
function evaluations} (also known as \emph{value oracle calls}), it is 
typically measured (as in this work) by their number. So, here the goal is to 
design algorithms requiring an almost linear number of such evaluations.
There is extensive research focusing on this issue, be it in the standard 
algorithmic setting \citep{MirzasoleimanBK16}, or in streaming 
\citep{ChekuriGQ15,BadanidiyuruMKK14} and distributed submodular 
maximization \citep{MirzasoleimanKSK13,BarbosaENW15}. 
The second challenge is the inherent uncertainty in  problems like sensor 
placement or revenue maximization, where one does not learn  the exact 
marginal value of an element until it is added to the solution (and thus ``paid 
for''). This, too, has motivated several works on adaptive submodular 
maximization \citep{GolovinK11,GotovosKK15,GuptaNS17,Mitrovic0F0K19}. 
Note that even estimating the expected value to a partially unknown objective 
function can be very costly and this makes the reduction of the number of such
calls all the more  important.

Knapsack constraints are one of the most natural types of restriction that occurs in  real-world problems and  are often \emph{hard} budget, time, or size constraints. Other combinatorial constraints like partition matroid constraints, on the other hand, model less stringent requirements, e.g., avoiding too many similar items in the solution. As the soft versions of such constraints can be often hardwired in the objective itself (see the \emph{Video Recommendation} application in \cref{sec:experiments}), we do not deal with them directly here.

The nearly-linear time requirement, without large constants involved, leaves little room for using sophisticated approaches like continuous greedy methods 
\citep{FeldmanNS11} or enumeration of initial solutions \citep{Sviridenko04}. To further highlight the delicate balance between function evaluations and 
approximation, it is worth mentioning that, even for the monotone case, the first result combining $O(n\log n)$ oracle calls with an approximation better than $2$ is the very recent $\frac{e}{e-1}$-approximation algorithm of 
\citet{EneN19icalp_a}. While this is a very elegant theoretical result, the huge constants involved render it unusable in practice.

At the same time, the strikingly simple, $2$-approximation \emph{modified 
densitiy greedy} algorithm of \citet{Wolsey82} deals well with both issues in the \emph{monotone} case:
\emph{Sort the items in decreasing order according to their marginal value over cost ratio and pick as many items as possible in that order without violating the constraint. Finally, return the best among this solution and the best single item.} 
When combined with lazy evaluations \citep{Minoux78}, this algorithm requires only $O(n\log n)$ value oracle calls and can be adjusted to work equally well for adaptive submodular maximization \citep{GolovinK11}. 
For \emph{non-monotone} objectives, however, the only practical algorithm is 
the $(10+\varepsilon)$-approximation FANTOM algorithm of 
\citet{MirzasoleimanBK16} requiring $O(n^2 \log n)$ value oracle calls. 
Moreover,  there is no known algorithm for the adaptive setting that can handle anything beyond a cardinality constraint \citep{GotovosKK15}.

We aim to tackle both aforementioned challenges for non-monotone submodular 
maximization under a knapsack constraint, by revisiting the simple algorithmic principle of Wolsey's density greedy algorithm. Our approach is along the lines of recent results on  \emph{random greedy} combinatorial algorithms 
\citep{BuchbinderFNS14,FeldmanHK17}, which show that introducing 
randomness into greedy algorithms can extend their guarantees to the non-monotone case. Here, we give the first such algorithm for a knapsack constraint.

\subsection{Contribution and Outline}

The density greedy algorithm mentioned above may produce arbitrarily poor 
solutions when the objective is non-monotone.
In this work we show that introducing some randomization leads to  a simple 
algorithm, \textsc{SampleGreedy}, that outperforms existing algorithms both in theory and in practice. \textsc{SampleGreedy}  flips a coin before greedily 
choosing any item in order to decide whether to include it to the solution or 
ignore it. The algorithmic simplicity of such an approach keeps \textsc{SampleGreedy} fast, easy to implement, and flexible enough to adjust to other related settings. At the same time the added randomness prevents it from getting trapped in solutions of poor quality. 

In particular, in \cref{sec:vanilla_greedy} we show that 
\textsc{SampleGreedy} is a $5.83$-approximation algorithm using only $O(n \log n)$ value oracle calls.
When all singletons have small value compared to an optimal solution, the 
approximation factor improves to almost $4$.  This is the first constant-factor approximation algorithm for the non-monotone case using this few queries. The only other algorithm fast enough to be suitable for large instances is the aforementioned FANTOM  \citep{MirzasoleimanBK16} which, for a knapsack constraint,\footnote{FANTOM can handle more general constraints, like a $p$-system constraint and $\ell$ knapsack constraints. Here we refer to its performance and running time when restricted to a single knapsack constraint.} achieves an approximation factor of $(10+\varepsilon)$ with  $O(n r \varepsilon^{-1}\log n)$ queries, where $r$ is the size of the largest 
feasible set and can be as large as $\Theta(n)$. Even if we modify FANTOM to 
use lazy evaluations, we still improve the query complexity by a logarithmic 
factor (see also \cref{rem:lazyfantom}).

For the adaptive setting, where the stochastic submodular objective  is learned as we build the solution, we show in \cref{sec:adaptive_greedy} that a
variant of our algorithm, \textsc{AdaptiveGreedy}, still guarantees a 
$9$-approximation to the \emph{best adaptive policy}. This is not only a 
relatively small loss given the considerably stronger benchmark, but is in fact the first constant approximation known for the problem in the 
adaptive submodular maximization framework of \citet{GolovinK11} and 
\citet{GotovosKK15}. 
Hence we fill a notable theoretical gap, given that models with incomplete prior information or those capturing evolving settings are becoming increasingly important in practice.

From a technical point of view, our algorithm combines the simple principle of
always choosing a high-density item with maintaining a careful 
exploration-exploitation balance, as is the case in many stochastic learning 
problems. It is therefore directly related to the recent simple randomized greedy approaches for maximizing non-monotone submodular objectives subject to
other (i.e., non-knapsack) constraints \citep{BuchbinderFNS14,ChekuriGQ15,FeldmanHK17}. However, there are 
underlying technical difficulties that make the analysis for knapsack constraints significantly more challenging. Every single result in this line of work critically depends on making a random choice in each step, in a way so that ``good progress'' is consistently made. This is not possible under a knapsack constraint. Instead, we argue globally about the value of the 
\textsc{SampleGreedy}  output via a  comparison with a carefully maintained 
\emph{almost integral} solution. When it comes to extending this approach to 
the adaptive non-monotone submodular maximization framework, we crucially 
use the fact that the algorithm builds  the solution iteratively, committing in every step to all the past choices. This is the main technical reason why it is not possible to adjust algorithms with multiple ``parallel'' runs, like FANTOM, to the adaptive setting.

Our algorithms provably handle well the aforementioned emerging, modern-day 
challenges, i.e., stochastically evolving objectives and rapidly growing 
real-world instances. In \cref{sec:experiments} we showcase the fact that 
our theoretical results indeed translate into applied performance. We focus on 
two applications that fit within the framework of non-monotone submodular 
maximization subject to a knapsack constraint, namely \emph{video 
recommendation} and \emph{influence-and-exploit marketing}. We run 
experiments on real and synthetic data that indicate that  \textsc{SampleGreedy} consistently performs better than FANTOM while being much faster. For \textsc{AdaptiveGreedy} we highlight the fact that its adaptive behavior results in a significant improvement over non-adaptive alternatives.

\subsection{Related Work}
There is an extensive literature on submodular maximization subject to knapsack or other constraints, going back several decades, see, e.g., 
\citep{NemhauserWF78,Wolsey82}.
For a \textit{monotone} submodular objective subject to a knapsack constraint 
there is a deterministic  $\frac{e}{e-1}$-approximation algorithm 
\citep{khuller1999budgeted,Sviridenko04} which is tight, unless 
$\text{P}=\text{NP}$ \citep{Feige98}. 

On non-monotone submodular functions \citet{LeeMNS10} provided a 
$5$-approximation algorithm for $k$ knapsack constraints, which was the first 
constant factor algorithm for the problem.  \citet{FadaeiFS11} building on the 
approach of \citet{LeeMNS10}, reduced this factor to $4$.
One of the most interesting algorithms for a single knapsack constraint is the 
$6$-approximation algorithm of \citet{GuptaRST10}. As this is a greedy 
combinatorial algorithm based on running Sviridenko's algorithm twice, it is 
often used as a subroutine by other algorithms in the literature, e.g., 
\citep{BarbosaENW15}, despite its running time of $O(n^4)$. A number of 
continuous greedy approaches \citep{FeldmanNS11,KulikST13,ChekuriVZ14} led 
to the current best factor of $e$ when a knapsack---or even a general 
downwards closed---constraint is involved. However, continuous greedy 
algorithms are impractical for most real-world applications. The fastest such 
algorithm for our setting is the $(e+\varepsilon)$-approximation algorithm of 
\citet{ChekuriJV15} requiring $O(n^3\varepsilon^{-4}\mathrm{polylog}(n))$ 
function evaluations. Possibly the only algorithm that is directly comparable to 
our \textsc{SampleGreedy} in terms of running time is FANTOM by 
\citet{MirzasoleimanBK16}. FANTOM achieves a 
$(1+\varepsilon)(p+1)(2p+2\ell+1) / p$-approximation for $\ell$ knapsack 
constraints and a $p$-system constraint in time $O(nrp\varepsilon^{-1}\log(n))$, 
where $r$ is the size of the largest feasible solution.

As mentioned above, there is a number of recent results on randomizing simple 
greedy algorithms so that they work for non-monotone submodular objectives 
\citep{BuchbinderFNS14,ChekuriGQ15,GotovosKK15,FeldmanHK17,FeldmanZ18}.
 Our paper extends this line of work, as we are the first to successfully apply 
this approach for a knapsack constraint. 

\citet{GolovinK11} introduced the notions of adaptive monotonicity and 
submodularity and showed it is possible to achieve guarantees with respect to 
the optimal adaptive policy that are similar to the guarantees one gets in the 
standard algorithmic setting with respect to an optimal solution. 
Our \cref{sec:adaptive_greedy} fits into this framework as it was 
generalized by \citet{GotovosKK15} for non-monotone objectives; they showed that a variant of the random greedy algorithm of 
\citet{BuchbinderFNS14}  achieves a $\frac{e}{e-1}$-approximation in the case of 
a cardinality constraint.

Implicitly related to our quest for few value oracle calls is the recent line of work 
on the adaptive complexity of submodular maximization that measures the 
number of sequential rounds of independent value oracle calls needed to obtain 
a constant factor approximation; see 
\citep{BalkanskiS18,BalkanskiRS19,FahrbachMZ19a,FahrbachMZ19b} and 
references therein.
To the best of our knowledge, nothing nontrivial is known for non-monotone 
functions and  a knapsack constraint.

\section{Preliminaries}
In this section we formally introduce the problem of submodular maximization 
with a knapsack constraint in both the standard and the adaptive setting.

Let  $A = \{1, 2, ..., n\}$ be a set of $n$ items.
\begin{definition}[Submodularity]\label{def:SM}
	A  function $v: 2^A \rightarrow \mathbb{R}$ is 
	\emph{submodular} if and only if 
	$$v(S\cup \{i\}) - v(S) \geq v(T\cup \{i\}) - v(T)$$ for all $S\subseteq T 
	\subseteq 
	A$ and 
	$i\not\in T$.
\end{definition}
A  function $v : 2^A \rightarrow \mathbb{R}$ is \emph{non-decreasing} (often 
referred to as 
\emph{monotone}), if $v(S) \le v(T)$ for any $S \subseteq T \subseteq A$. We 
consider general (i.e., not necessarily monotone), normalized (i.e., 
$v(\emptyset)=0$), non-negative submodular valuation functions.
Since \emph{marginal} values are extensively used, we adopt the shortcut 
$v(T\,|\,S)$ for 
the marginal value of set $T$ with respect to a set $S$, i.e. $v(T\,|\,S)=v(T \cup 
S) - v(S).$ If $T = \{i\}$ we write simply $v(i\,|\,S)$.  While this is the most 
standard definition of submodularity in this setting, there 
are alternative equivalent definitions that will be useful later on.
\begin{theorem}[\citet{NemhauserWF78}]\label{thm:SM}
	Given a  function $v:2^A \rightarrow \mathbb{R}$, the 
	following are equivalent:
	\begin{enumerate}
		\item $v(i\,|\,S) \geq v(i\,|\,T)$ for all $S\subseteq T \subseteq A$ and 
		$i\not\in T$. \label{def:sub1}
		\item $v(S) + v(T) \ge v(S\cup T) + v(S\cap T)$ for all $S, T \subseteq A$. 
		\label{def:sub2}
		\item $v(T) \le v(S) + \sum_{i \in T \setminus S} v(i\,|\,S) - \sum_{i \in S 
		\setminus T} v(i\,|\,S\cup T \setminus \{i\})$ for all $S, T \subseteq 
		A$.\label{def:sub3} 
	\end{enumerate}
\end{theorem}
Moreover, we restate a key result which connects random sampling and 
submodular maximization. The original version of the theorem was due to 
\citet{FeigeMV11}, although here we use a variant from \citet{BuchbinderFNS14}. 
\begin{lemma}[Lemma 2.2. of \citet{BuchbinderFNS14}]
	\label{lem:sampling}
	Let $f: 2^A \to \mathbb{R}$ be a submodular set function, let $ X \subseteq 
	A$ and let $X(p)$ be a sampled subset, where each element of $X$ appears 
	with probability at most $p$ (not necessarily independent). Then:
	$$
	\mathbb{E} \left[ f(X(p)) \right] \geq (1-p)f(\emptyset).
	$$
\end{lemma}

We assume access to a \emph{value oracle} 
that 
returns $v(S)$ when given as input a set $S$.
We also associate a positive cost $c_i$ with each element $i \in A$ 
and consider a given 
budget $B$. The goal is to find a subset of $A$ of maximum value among the 
subsets whose total cost is at most $B$.
Formally, we want some $S^* \in \argmax\{v(S)\,|\, S \subseteq A, \, \sum_{i \in 
S} c_i \leq B\}$. 
Without loss of generality, we may assume that $c_i\le B$ for all $i\in A$, since 
any element with cost exceeding $B$ is not contained in any feasible solution 
and can be discarded.

We next present the adaptive optimization framework. 
On a high level, here we do know how the world works and what situations 
occur with which probability. However, which of those we will be actually 
dealing with is inferred over time by the bits of information we learn.
Along with set $A$, we introduce the \emph{state space} $\Omega$ which is 
endowed with some probability measure. By $\omega = (\omega_i)_{i \in A} \in 
\Omega$ we  specify the \textit{state} of each element in $A$. The adaptive 
valuation function $v$ is then defined over $A \times \Omega$; the value over a 
subset $S \subseteq A$ depends on both the subset and $\omega$. 
Due to the probability measure over $\Omega$, $v(S,{\omega})$ is a random 
variable. We define $v(S)= \E{v(S,{\omega})}$, the expectation being with respect 
to ${\omega}$. Like before, the costs $c_i$ are deterministic and known in 
advance.

For each $\omega \in \Omega$ and $S \subseteq A$, we define the partial 
realization of state $\omega$ on $S$ as the couple $(S, \omega_{|S})$, where 
$\omega_{|S}= (\omega_i)_{i \in S}.$ It is natural to assume that the true value of 
a set $S$ does not depend on the whole state, but only on $\omega_{|S}$, i.e., 
$v(S, \omega)= v(S, \psi),$ for all $\omega, \psi \in \Omega$ such that 
$\omega_{|S}=\psi_{|S}$. Therefore, sometimes we overload the notation and 
use $v(S,\omega_{|S})$ instead of $v(S,\omega)$.
There is a clear partial ordering on the set of all possible partial realizations: 
$(S,\omega_{|S}) \subseteq (T,\omega_{|T})$ if $S \subseteq T$ and 
$\omega_{|T}$ coincides with $\omega_{|S}$ over all the elements of $S.$
The marginal value of an element $i$ given a partial realization $(S,\omega_{|S})$
is 
\[
v(i\,|\,(S,\omega_{|S})) = \E{v(\{i\} \cup S, {\omega}) - v(S, 
{\omega})\,|\,{\omega}_{|S}}.
\]
We are now ready to introduce the concepts of adaptive submodularity and 
monotonicity.
\begin{definition}
	The function $v(\cdot\,,\cdot)$ is 
	\emph{adaptive submodular} if $v(i\,|\,(S,\omega_{|S})) \geq 
	v(i|(T,\omega_{|T}))$ for all partial realizations $(S, \omega_{|S}) \subseteq (T, 
	\omega_{|T})$ and  for any $i \notin T$.
	Further, $v(\cdot\,,\cdot)$ is \emph{adaptive monotone} if 
	$v(i\,|\,(S,\omega_{|S})) \geq 0$ for all partial realizations $(S, \omega_{|S})$ 
	and for all $i \notin S$.
\end{definition}

In \cref{sec:adaptive_greedy} we assume access to a value oracle that 
given an element $i$ and a partial realization returns the expected marginal value 
of $i$.
Using the properties of conditional expectation, it is straightforward to show that 
if $v(\cdot , \cdot)$ is adaptive submodular, then its expected value $v(\cdot)$ 
is submodular. In analogy with \citep{GotovosKK15}, we assume $v$ to be 
state-wise submodular, i.e., $v(\cdot, \omega)$ is a submodular set function for 
each $\omega \in \Omega.$

In this framework it is possible to define \textit{adaptive policies} to maximize 
$v$. An adaptive policy is a function
which associates with every partial realization a distribution on the next element 
to be added to the solution. 
The optimal solution to the adaptive submodular maximization problem is to find 
an adaptive policy that maximizes the expected value while respecting the 
knapsack constraint (the expectation being taken over $\Omega$ and the 
randomness of the policy itself).

\section{The Algorithmic Idea} 
\label{sec:vanilla_greedy}

We present and analyze \textsc{SampleGreedy}, a randomized 
$5.83$-approximation algorithm for maximizing a submodular function subject to 
a knapsack constraint. As we mentioned already, \textsc{SampleGreedy} is 
based on the modified density greedy algorithm of \citet{Wolsey82}. Since the 
latter may perform arbitrarily bad for non-monotone objectives, we add a 
sampling phase, similar to the sampling phase of the Sample Greedy algorithm 
of \citet{FeldmanHK17}.

\textsc{SampleGreedy} first selects a random subset $A'$ of $A$ by 
independently 
picking each element with probability $p$. Then it runs Wolsey's algorithm only 
on $A'$. To formalize this second step, using $v(i)$ as a shorthand for $v(\{i\})$, 
let $$j_{1} \in \argmax_{i\in A'} v(i) / c_i,$$ while, for $k\ge 1$, $$j_{k+1} \in \argmax_{i\in A'\mysetminus\{j_1,\ldots,j_k\}}  v(i\,|\,\{j_1, \ldots, 
j_k\}) 
/ c_i.$$ 
If $\ell$ is the largest integer such that $\sum_{i=1}^{\ell} c_{j_i} \le B$, then 
$S=\{j_1,\ldots,j_{\ell}\}$. In the end, the output is the one yielding the largest 
value between $S$ and an element from $\argmax_{i\in A'} v(i)$.

We formally present this algorithm in pseudocode below. Notice that to simplify 
the  analysis, instead of selecting the entire set $A'$ immediately at the start of 
the 
algorithm, we defer this decision and toss a coin with probability of success $p$ 
each time an item is considered to be added to the solution. Both versions of 
the algorithm behave identically.
\begin{algorithm}[ht!]
	\DontPrintSemicolon 
	\NoCaptionOfAlgo
	\algotitle{\textnormal{\textsc{SampleGreedy}}}{fig:SampleGreedy.title}
	$i^* \in \argmax_{k\in A} v\left( k \right) $ \tcc*{{\scriptsize best single item}}
	$S= \emptyset$ \tcc*{{\scriptsize greedy solution}} \label{line:S_initial}
	$F=\left\{k\in A\,|\,v(k\,|\,S)>0 \text{ and } c_k\le B\right\}$  \tcc*{{\scriptsize 
	initial set of feasible items}}
	$R = B$ \tcc*{{\scriptsize remaining knapsack capacity}} \label{line:R_budget}
	\While{$F\neq \emptyset$}{
		Let $i \in \argmax_{k\in F} \displaystyle\frac{v(k\,|\,S)}{c_k}$ 
		\label{line:argmax}\;
		Let $r_i\sim \mathrm{Bernoulli}(p)$ \tcc*{{\scriptsize independent random 
		bit}} 	
		\If{$r_i = 1$}{
			$S = S \cup \{i\}$ \;
			$R = R - c_i$\;
		}
		$A = A\setminus\{i\}$ \;
		$F=\left\{k\in A\,|\,v(k\,|\,S)>0 \text{ and } c_k\le R\right\}$ 
		\label{line:F_update}
	}
	\Return $\max\{v(i^*), v(S)\}$
	\caption{\textsc{SampleGreedy}$(A, v, \mathbf{c}, B)$} 
	\label{fig:SampleGreedy} 
\end{algorithm}

\begin{theorem}
	\label{thm:greedy}
	For $p=\sqrt{2} - 1$, \textsc{SampleGreedy} is a  $\left(3+2\sqrt{2} 
	\right)$-approximation algorithm. 
\end{theorem}
\begin{proof}
	For the analysis of the algorithm we are going to use the auxiliary set $O$, an 
	extension of the set $S$ that respects the knapsack constraint and uses 
	feasible items from an optimal solution. In particular, let $S^*$ be an optimal 
	solution and let $s_1, s_2, \ldots, s_{r}$ be its elements. 
	
	Then, $O$ is a \emph{fuzzy} set that is initially equal to $S^*$ and 
	during each iteration of the \emph{while} loop it is updated as follows:
	\begin{itemize}
		\item If $r_i = 1$, then $O = O \cup \{i\}$. In case this addition violates the 
		knapsack constraint, i.e., $c(O)>B$, then we repetitively remove items from 
		$O\setminus S$ in increasing order with respect to their cost until the cost 
		of $O$ becomes exactly $B$. Note that this means that the last item 
		removed may be removed only partially. More precisely, if $c(O)>B$ and 
		$c(O\setminus \{s_j\}) \le B$, where  $s_j$ is the item of $S^*$ of maximum 
		index in $O\setminus S$, then we keep a $\big( B - c( O )  + c_j \big)\,  /c_j$ 
		fraction of $s_j$ in $O$ and stop its update for the current iteration.
		\item Else (i.e., if $r_i = 0$), $O = O \setminus \{i\}$.
	\end{itemize}
	If an item $j$ was considered (in \textcolor{links}{line} \ref{line:argmax}) in some iteration of the 
	\emph{while} loop, then let $S_j$ and $O_j$ denote the sets $S$ and $O$, 
	respectively, at the beginning of that iteration. Moreover, let $O'_j$ denote 
	$O$ at the end of that iteration. If $j$ was never considered, then $S_j$ and 
	$O_j$ (or $O'_j$)  denote the final versions of $S$ and $O$, respectively. In 
	fact, in what follows we exclusively use $S$ and $O$ for their final versions.
	
	It should be noted that, for all $j\in A$, $S_j\subseteq O_j$ and also no item 
	in $O_j\setminus S_j$ has been considered in any of the previous iterations of 
	the \emph{while} loop.

	Before stating the next lemma, let us introduce some notation for the sake of 
	readability. Note that, by construction, $O\setminus S$ is either empty or 
	consists of a single fractional item $\hat \imath$. In case $O\setminus S = 
	\emptyset$, by $\hat \imath$ we denote the last item removed from $O$. For 
	every $j\in A$, we define $Q_j = O_j \setminus (O_j' \cup S 
	\cup\{\hat\imath\,\})$. Note that if $j$ was never considered during the 
	execution of the algorithm, then $Q_j =\emptyset$.
	
	\begin{lemma}\label{lem:O_vs_SUS*}
		For every  realization of the Bernoulli random variables, it holds that 
		\[v(S\cup S^*) \le v(S) + v(\hat\imath\,) + \displaystyle\sum_{j\in A} c(Q_j)\, 
		\frac{v( j \,|\, S_j) }{c_{j}}.\]
	\end{lemma}
	\begin{proof}[Proof of \cref{lem:O_vs_SUS*}]
		Assume that the random bits $r_1, r_2, \ldots$ are fixed.
		Also, without loss of generality, assume the items are numbered according 
		to the order in which they are considered by \textsc{SampleGreedy}, with 
		the ones not considered by the algorithm numbered arbitrarily (but after the 
		ones considered). That is, item $j$---if considered---is the item considered 
		during the $j^{th}$ iteration. 
		
		Consider now any round $j$ of the while loop of \textsc{SampleGreedy}. 
		An item is removed from $O_j$ in two cases. First, it could be item $j$ 
		itself that was originally in $S^*$ but $r_j=0$ (and hence it will never get 
		back in $O_k$ for any $k > j$). Second, it could be some other item that 
		was in $S^*$ and is taken out to make room for the new item $j$. In the 
		latter case the only possibility for the removed item to return in $O_k$ for 
		some $k>j$ is to be selected by the algorithm and inserted in $S$. We can 
		hence conclude that $Q_j \cap Q_k = \emptyset$ for all $j \neq k.$
		In addition to that, it is clear that $S \cup S^* = S 
		\cup\{\hat\imath\,\}\cup_{j=1} Q_j.$
		
		Therefore, if items $1, 2, \dots, \ell$ where all the items ever considered, 
		using  submodularity and the fact that $S_j \subseteq S \subseteq S 
		\cup_{r = j+1}^{\ell} Q_r $, we have
		\begin{eqnarray}\label{eq:O_vs_SUS*}
		\nonumber
		v(S \cup S^*) - v(S) - v(\hat\imath\,) & \le & v((S \cup S^*) 
		\setminus\hat\imath\,) - v(S) = \sum_{j=1}^{\ell} v\left(Q_j\,\big|\, S \cup_{r = 
		j+1}^{\ell} Q_r\right) \\ \nonumber &\leq& \sum_{j=1}^{\ell} v\left(Q_j\,|\, S_j 
		\right) \leq \sum_{j=1}^{\ell} \sum_{x \in Q_j} \frac{v(x\,|\,S_j)}{c_x}\cdot c_x  
		\\\nonumber
		&\leq &  
		\sum_{j=1}^{\ell} \sum_{x \in Q_j} \frac{v(j\,|\,S_j)}{c_j}\cdot c_x = 
		\sum_{j=1}^{\ell} \frac{v(j\,|\,S_j)}{c_j}\cdot \sum_{x \in Q_j}c_x \\ & = 
		&\sum_{j=1}^{\ell} \frac{v(j\,|\,S_j)}{c_j}\cdot c(Q_j)\,,
		\end{eqnarray}
		where in a slight abuse of notation we consider $c_x$ to be the fractional 
		(linear) cost if $x \in Q_j$ is a fractional item. While the first three 
		inequalities directly follow from the submodularity of $v$, for the last 
		inequality we need to combine  the optimality of $v(j\,|\,S_j)/c_j$ at the step 
		$j$ was selected with the fact that every single item $x$ appearing in the 
		sum $\sum_{j=1}^{\ell} \sum_{x \in Q_j} v(x\,|\,S_j)$ was feasible (as a whole 
		item) at that step. The latter is true because of the way we remove items 
		from $O$. If $x$ is removed, it is removed before (any part of) $\hat\imath$ 
		is removed. Thus, $x$ is removed when the available budget is still at least 
		$c_{\hat\imath}$. Given that $c_x \le c_{\hat\imath}$, we get that $x$ is 
		feasible until removed. 
		
		To conclude the proof of the Lemma it is sufficient to note that $c(Q_j)=0$ 
		for all items that were not considered.
	\end{proof}
	
	While the previous Lemma holds for \textit{each} realization of the random 
	coin tosses in the algorithm, we next consider inequalities holding in 
	expectation over the randomness of the $\{r_i\}_{i=1}^{|A|}$ in 
	\textsc{SampleGreedy}. The indexing of the elements is hence to be 
	considered deterministic and fixed in advance, not as in the proof of 
	\cref{lem:O_vs_SUS*}.
	
	\begin{lemma}\label{lem:S_vs_sum}
		$\displaystyle\E{\sum_{j\in A} c(Q_j)\, \frac{v( j \,|\, S_j) }{c_{j}}} \le 
		\frac{\max\{p, 1-p\}}{p}\,\E{v(S)}$ 
	\end{lemma}
	\begin{proof}[Proof of \cref{lem:S_vs_sum}]
		For all $i \in A$, we define $G_i$ to be the random gain because of $i$ at 
		the time $i$ is added to the solution ($G_i = v(i\,|\,S_i) $ if $i$ is added and 
		$0$ otherwise)
		
		Since $v(S) = \sum_{i \in A}G_i$, by linearity, it suffices to show that the 
		following inequalities hold in expectation over the coin tosses:
		\begin{equation}\label{eq:decouple}
		c(Q_i) \, \frac{v(i\,|\,S_i)}{c_i} \leq \frac{\max\{p, 1-p\}}{p} \,G_i, \quad \forall i 
		\in A.
		\end{equation}
		In order to achieve that, following \citep{FeldmanHK17}, let $\Ei$ be any 
		event specifying the random choices of the algorithm up to the point $i$ is 
		considered (if $i$ is never considered, $\Ei$ captures all the randomness).
		If $\Ei$ is an event that implies $i$ is not considered, then the 
		\cref{eq:decouple} is trivially true, due to $G_i=0$ and $Q_i = \emptyset$.
		We focus now on the case $\Ei$ implies that $i$ is considered. Analyzing 
		the algorithm, it is clear that
		\begin{equation}
		\E{c(Q_i)\,|\,\Ei} \le 
		\begin{cases}
		0 \cdot \p{r_i = 1} + c_i \cdot \p{r_i = 0} = (1-p) \cdot c_i\,, \quad &\text{if $i 
		\in O_i$}\,,\\
		c_i  \cdot \p{r_i = 1} + 0 \cdot \p{r_i = 0} = p \cdot c_i\,, \quad \quad 
		&\text{otherwise}.
		\end{cases}
		\end{equation}
		In short, $\E{c(Q_i)\,|\,\Ei} \leq \max\{p, 1-p\} \cdot c_i.$ It is here that we 
		use the fuzziness of $O$: without the fractional items it would be hopeless 
		to bound $c(Q_t)$ with $c_t$.
		
		At this point, we exploit the fact that $\Ei$ contains the information on 
		$S_{i}$, i.e., $S_{i}=S_{i}(\Ei)$ deterministically. Recall that $S_i$ is the 
		solution set at the time item $i$ is considered by the algorithm.

		\begin{align*}
		\E{G_i\,|\,\Ei} &= \p{i \in S\,|\,\Ei}v(i\,|\,S_i) = \p{r_i = 1}v(i\,|\,S_i) = p \cdot 
		c_i\, \frac{v(i\,|\,S_i)}{c_i}\\ &  \geq \frac{p}{\max\{p, 1-p\}} \, \E{c(Q_i)\,|\,\Ei} 
		\, \frac{v(i\,|\,S_i)}{c_i}\\ 
		& = \frac{p}{\max\{p, 1-p\}} \, \E{c(Q_i)\, \frac{v(i\,|\,S_i)}{c_i} \,\Big|\,\Ei}.
		\end{align*}
		We can hence conclude the proof by using the law of total probability over 
		$\Ei$ and the monotonicity of the conditional expectation:
		\begin{align*}
		\E{G_i} = \E{\E{G_i\,|\,\Ei}} &\geq \E{\frac{p}{\max\{p, 1-p\}} \,\E{c(Q_i)\, 
		\frac{v(i\,|\,S_i)}{c_i} \,\Big|\,\Ei}} = \\ &= \frac{p}{\max\{p, 1-p\}} \,\E{c(Q_i) \, 
		\frac{v(i\,|\,S_i)}{c_i}}.
		\end{align*}
	\end{proof}

	\begin{lemma}\label{lem:Feige+}
		$\displaystyle v(S^*)\le \frac{1}{1-p}\,\E{v(S\cup S^*)}$.
	\end{lemma}
	\begin{proof}[Proof of \cref{lem:Feige+}]
		Let $S^*$ be an optimal set for the constrained submodular maximization 
		problem. We define $g: 2^A \to \mathbb{R}_+$ as follows: $g(B)= v(B \cup 
		S^*).$ It is a simple exercise to see that such function is indeed 
		submodular, moreover $g(\emptyset)=v(S^*).$
		If we now apply \cref{lem:sampling} to $g$, observing that the elements in 
		the set $S$ output by the algorithm are chosen with probability at most 
		$p$, we conclude that:
		\[
		\E{v(S \cup S^*)} = \E{g(S)} \geq (1-p) g(\emptyset) = (1-p) v(S^*).
		\]
	\end{proof}
	
	Combining
	\textcolor{links}{Lemmata} \ref{lem:O_vs_SUS*}, \ref{lem:S_vs_sum} and \ref{lem:Feige+} we get
	\begin{align}
	\nonumber
	(1-p)v(S^*)  & \le \E{v(S\cup S^*)}\\
	\nonumber
	& \le \E{v(S) + v(\hat\imath\,)+\sum_{j\in A} c(Q_j)\, \frac{v( j \,|\, S_j) }{c_{j}}}\\
	\nonumber
	& \le  \E{v(S)} + v(i^*) + \frac{\max\{p, 1-p\}}{p}\,\E{v(S)} \\
	\label{eq:large_inst}
	& = \max\left\{2, \tfrac{1}{p}\right\} \cdot \E{v(S)} + v(i^*) \,.
	\end{align}
	By substituting $\sqrt{2} - 1$ for $p$, this yields $v(S^*) \le  
	(3+2\sqrt{2})\max\{\E{v(S)},v(i^*)\}$.
	This establishes the claimed approximation factor.
\end{proof}

A naive implementation of \textsc{SampleGreedy} needs $\Theta(n^2)$ value 
oracle calls in the worst case. Indeed, in each iteration all the remaining elements 
have their marginals updated and for large enough $B$ the greedy solution may 
contain a constant fraction of $A$. Applying lazy evaluations \citep{Minoux78}, 
however, we can cut the number of queries down to 
$O(n\varepsilon^{-1}\log{(n/\varepsilon)})$ losing only an $\varepsilon$ in the 
approximation factor (see also \citep{EneN19icalp_a}). To achieve this, instead 
of recomputing all the marginals at every step, we maintain an ordered queue of 
the elements sorted by their last known \emph{densities} (i.e., their marginal 
value per cost ratios) and use it to get a \textit{sufficiently good} element to add. 

More formally, the lazy implementation of \textsc{SampleGreedy} maintains the 
elements in a priority queue in decreasing order of density, which is initialised 
using the ratios $v(i) / c_i$. 
At each step we pop the element on top of the queue. 
If its density with respect to the current solution is within a $1+\varepsilon$ 
factor of its old one, then it is picked by the algorithm, otherwise it is reinserted 
in the queue according to its new density and we pop the next element. 
Submodularity guarantees that the density of a picked element is at least 
$1/(1+\varepsilon)$ of the best density for that step. 
As soon as an element has been updated $\log(n/\varepsilon)/\varepsilon$ times, 
we discard it. 

\begin{theorem}\label{thm:lazygreedy}
	The lazy version of \textsc{SampleGreedy} achieves an approximation factor 
	of $3+2\sqrt{2}+\varepsilon$
	using $O(n\varepsilon^{-1}\log{(n/\varepsilon)})$ value oracle calls.
\end{theorem}
\begin{proof}
	For a given $\varepsilon\in (0,1)$ let $\varepsilon' = \varepsilon/6$. We 
	perform lazy evaluations using $\varepsilon'$, with $\log$
	denoting the binary logarithm.
	
	It is straightforward to argue about the number of value oracle calls. Since the 
	marginal value of each element $i$ has been updated at most 
	$\frac{\log(n/\varepsilon')}{\varepsilon'}$ times, we have a total of at most 
	$n\frac{\log(n/\varepsilon')}{\varepsilon'} = 
	O\big(\frac{n\log(n/\varepsilon)}{\varepsilon}\big)$ function evaluations.
	
	The approximation ratio is also easy to show. There are two distinct sources 
	of loss in approximation. We first bound the total value of the discarded 
	elements due to too many updates. This value appears as the upper bound of 
	an extra additive term in \cref{eq:O_vs_SUS*}. Indeed, now besides 
	$\sum_{j=1}^{\ell} v\left(Q_j\,\big|\, S \cup_{r = j+1}^{\ell} Q_r\right)$ we need 
	to account for the elements of $O$ that were ignored because of too many 
	updates. Such elements, once they become ``inactive'' do not contribute to 
	the cost of the current $O$ and are never pushed out as new elements come 
	into $S$. The definition of the $Q_j$s in the proof of \cref{thm:greedy} 
	should be adjusted accordingly.
	That is, if $W_j$ are the elements of $O$ that become inactive because they 
	were updated too many times during iteration $j$, we have
	\[v((S \cup S^*) \setminus\hat\imath\,) - v(S) \le \sum_{j=1}^{\ell} 
	v\left(Q_j\,\big|\, S_j\right) + \sum_{j=1}^{\ell} v\left(W_j\,\big|\, S_j\right).\]
	However, by noticing that for $x\in (0,1)$ it holds that $x\le \log (1+x)$, we 
	have
	\begin{eqnarray}
	\nonumber
	\sum_{j=1}^{\ell} v\left(W_j\,\big|\, S_j\right) & \le & \sum_{i\in \bigcup_j W_j} 
	(1+\varepsilon')^{-\frac{\log(n/\varepsilon')}{\varepsilon'}}v(i)  \\ \nonumber 
	&\leq& \sum_{i\in \bigcup_j W_j} 
	(1+\varepsilon')^{-\frac{\log(n/\varepsilon')}{\log(1+\varepsilon')}} \max_{k\in A} 
	v(k) \\\nonumber
	&\le  &  
	\sum_{i\in A} (1+\varepsilon')^{-\log_{1+\varepsilon'}(n/\varepsilon')} v(S^*) \\ 
	\nonumber
	& = &\sum_{i\in A} \frac{\varepsilon}{6n} v(S^*) = \frac{\varepsilon}{6} v(S^*)\,.
	\end{eqnarray}

	For the second source of loss in approximation, recall that the marginals only 
	decrease due to submodularity. So, we know that if some item $j$ is 
	considered during iteration $j$ (following the renaming of 
	\cref{lem:O_vs_SUS*}), then $(1+\varepsilon')v(j\,|\, S_j)/c_j \ge \argmax_{k\in F} 
	v(k\,|\, S_j)/c_k$. The only difference this makes (compared to the  proof of 
	\cref{thm:greedy}) is that in the last inequality of \cref{eq:O_vs_SUS*} 
	we have an extra factor of $1+\varepsilon'$.
	
	Combining the above, we get the following analog of
	\cref{lem:O_vs_SUS*}:
	\[v(S\cup S^*) \le v(S) + v(\hat\imath\,) + \frac{\varepsilon}{6} v(S^*)+ 
	\displaystyle\sum_{j\in A} \Big(1+\frac{\varepsilon}{6}\Big) c(Q_j)\, \frac{v( j \,|\, 
	S_j) }{c_{j}},\]
	which carries over to \cref{eq:large_inst}, while \textcolor{links}{Lemmata} \ref{lem:S_vs_sum} 
	and \ref{lem:Feige+} are not affected at all. It is then a matter of simple 
	calculations to see that for  $p=\sqrt{2} - 1$, we still get $v(S^*) \le  
	(3+2\sqrt{2}+\varepsilon)\max\{\E{v(S)},v(i^*)\}$.
\end{proof}

Additionally, our analysis implies that \textsc{SampleGreedy} performs 
significantly better in the \emph{large instance} scenario, i.e., when the value of 
the optimal solution is much larger than the value of any single element. While it 
is not expected to have  exact knowledge of the factor $\delta$ in the following 
proposition, often some estimate is accessible. Especially for massive 
instances, it is reasonable to assume that $\delta$ is bounded by a very small 
constant.

\begin{theorem}\label{thm:largemarketgreedy}
	If $\max_{i \in A}v(i)\leq \delta \cdot \textsc{opt}$ for $\delta \in (0,1/2)$, then 
	\textsc{SampleGreedy} with $p= \frac{1-\delta}{2}$ is a $\left( 4 + 
	\varepsilon_\delta\right)$-approximation algorithm, where $\varepsilon_\delta 
	= \frac{4\delta (2-\delta)}{(1-\delta)^2}.$ 
\end{theorem}
\begin{proof}
	Starting from \cref{eq:large_inst} and exploiting the large instance property, 
	we get:
	\begin{align*}
	(1-p)v(S^*) \leq \max\left\{2, \tfrac 1p \right\} \cdot \E{v(S)} + v(i^*) \leq 
	\max\left\{2, \tfrac 1p \right\} \cdot \E{v(S)} + \delta \cdot v(S^*) . 
	\end{align*}
	
	Rearranging the terms and assuming $ p + \delta < 1,$ we have:
	$$
	v(S^*) \leq \frac{\max\left\{2, \tfrac 1p \right\}}{1-p- \delta}.
	$$
	Optimizing for $p \in (0,1-\delta)$ we get the desired statement.
\end{proof}

\section{Adaptive Submodular Maximization} 
\label{sec:adaptive_greedy}

In this section we modify \textsc{SampleGreedy} to achieve a good 
approximation guarantee in the adaptive framework.
Recall that the adaptive valuation function $v(\cdot\,,\cdot)$ depends on the 
state of the system which is discovered a bit at a time, in an adaptive fashion. 
Indeed, \textsc{SampleGreedy} is compatible with this framework and can be 
applied nearly as it is. 
We stick to the interpretation of \textsc{SampleGreedy} discussed right before \cref{thm:greedy}. That is, there is no initial sampling phase. Instead, 
we directly begin to choose greedily with respect the density (marginal value 
with respect to the current solution over cost). Each time we are about to pick an 
element of $A$, we throw a $p$-biased coin that determines whether we keep or 
discard the element. 

Here the main difference with the greedy part of \textsc{SampleGreedy} is that 
the marginals are to be considered with respect to the \emph{partial realization} 
relative to the current solution. Moreover, since it is not possible to return the 
largest between $\max_{i\in A}v(i)$ and the result of the greedy exploration, the 
choice between these two quantities has to be settled before starting the 
exploration. Formally, at the beginning of the algorithm a $p_0$-biased coin is 
tossed to decide between the two. 
The pseudo-code for the resulting algorithm, \textsc{AdaptiveGreedy}, is given
below. 

\begin{algorithm}[ht!]
	\DontPrintSemicolon 
	\NoCaptionOfAlgo
	\algotitle{\textnormal{\textsc{AdaptiveGreedy}}}{fig:AdaptiveSampleGreedy.title}
	Let $r_0 \sim \mathrm{Bernoulli}(p_0)$\;
	\If{$r_0=1$}{
		$i^* \in \argmax_{k\in A} v\left( k \right) $ \tcc*{{\scriptsize best single item 
		in expectation}}
		Observe $\omega_{i^*}$ and \Return{$v(i^*, \omega_{i^*})$}\;}
	$S= \emptyset, R = B$ \label{line:S=0} \tcc*{{\scriptsize greedy solution and 
	remaining knapsack capacity}}
	$F=\left\{k\in A\,|\,v(k)>0\right\}$\tcc*{{\scriptsize initial set of candidate items}}
	\While{$F\neq \emptyset$}{
		Let $i \in \argmax_{k\in F} \displaystyle \tfrac{v(k\,|\,(S, \omega_{|S}))}{c_k}$\;
		Let $r_i\sim \mathrm{Bernoulli}(p)$ \label{line:r_i-adaptive} \tcc*{{\scriptsize 
		independent random bit}} 	
		\If{$r_i = 1$}{
			Observe $\omega_i: \ S = S \cup \{i\}, \ R = R - c_i$\;
		}
		$A = A\setminus\{i\}, \quad F=\left\{k\in A\,|\,v(k\,|\,(S, \omega_{|S}))>0 \text{ 
		and } c_k\le R\right\}$
	}
	\Return $S, v(S, \omega_{|S})$ \label{line:return}
	\caption{\textsc{AdaptiveGreedy}
	} \label{fig:AdaptiveSampleGreedy}
\end{algorithm}\smallskip

Before  proving that \textsc{AdaptiveGreedy} works as promised, we need some 
observations. Let us denote by $S$ the output of a run of our algorithm, and 
$S^*$ the output of a run of the optimal adaptive strategy. Fix a realization 
$\omega \in \Omega$.
Now, Lemma 1 of \citep{GotovosKK15} implies 
\[ \mathbb{E}\left[ v(S\cup S^*, \omega) | \omega \right]\geq (1-p) \cdot v(S^*, 
\omega).
\]
Since $\omega$ (and therefore, $S^*$) is fixed, the only randomness is due to 
the coin flips in our algorithm. We stress that the union between $S$ and $S^*$ 
has to be intended in the following sense: run our algorithm, and independently, 
also the optimal one, both for the same realization $\omega$. The previous 
inequality is true for any $\omega$. So, by the law of total probability, we also 
have 
\begin{equation}\label{eq:adaptive_sampling}
\mathbb{E}\left[ v(S\cup S^*) \right] \geq (1-p) \cdot \mathbb{E}\left[ v(S^*)\right].
\end{equation}

For the next observation, assume our algorithm has picked (and therefore 
observed) exactly set $S$. That is, we know only $\omega_{|S}$.
We number all items $a\in A$ with positive marginal with respect to $S$ by 
decreasing ratio $\frac{v\left(a|(S,\omega_{|S})\right)}{c_a}$, i.e.,
\[a_1= \argmax_{a\in A} \left\{ \frac{v\left(a|(S,\omega_{|S})\right)}{c_a}\right\},\]
and so on. Note that this captures a notion of the \emph{best-looking items after 
already adding $S$}. 

For $k= \min\{i\in \mathbb{N}| \sum_{l=1}^i c_l\geq B\}$, we get, in analogy to 
Lemma 1 of \citet{GotovosKK15}, 
\begin{equation}
\label{eq:sum_vs_adap}
\sum_{i=1}^k v\left(a_i|(S,\omega_{|S})\right) \geq \E{\sum_{a\in S^*} v\left( a| 
(S,\omega_{|S})\right)| \omega_{|S}} \geq v\left(S^*|(S,\omega_{|S})\right). 
\end{equation}
Note that it could be the case that $k$ is not well defined, as there may not be 
enough elements with positive marginal to fill the knapsack. If that is the case, 
just consider $k$ to be the number of elements with positive marginals.

The point of \cref{eq:sum_vs_adap} is that, given $(S, \omega_|S)$, the set of 
elements $a_1 \dots a_k$ is deterministic, while $S^*$ is not, because it 
corresponds to the set opened by the best adaptive policy. Moreover, in the 
middle term notice that the conditioning influences the valuation, but \emph{not} 
the policy, since we are assuming to run it obliviously. This is fundamental for 
the analysis.

Since this holds for any set $S$, we can again generalize to the expectation over all possible runs of the 
algorithm (fixing the coin flips or not, as they only influence $S$; the best 
adaptive policy or the best-looking items $a_1, a_2, \dots, a_k$ are not 
affected). So, we get
\begin{equation}\label{two}
\mathbb{E}\left[\sum_{i=1}^k v\left(a_i|(S,\omega_{|S})\right) \right]
\geq \mathbb{E}\left[ v\left(S^*|(S,\omega_{|S})\right) \right].
\end{equation}
We remark that $k$ above is a random variable which depends on $S$.
We use these observations to prove the ratio of our algorithm.

\begin{theorem}
	\label{thm:adaptive}
	For $p_0 = 1/3$ and $p=1/6$,  \textsc{AdaptiveGreedy} yields a 
	$9$-approximation of $\textsc{opt}_{\Omega}$, while its
	\emph{lazy} version achieves a  $(9+\varepsilon)$-approximation using 
	$O(n\varepsilon^{-1}\log{(n/\varepsilon)})$ value oracle calls.
	
	Moreover, when $\max_{i \in A}v(i)\le \delta \cdot  \textsc{opt}_{\Omega} \ $ 
	for $\delta \in (0,1/2)$, then for $p_0=0$ and $p= (\sqrt{3-2\delta}-1)/2,$ 
	\textsc{AdaptiveGreedy} yields a $(4+2\sqrt{3} + 
	\varepsilon'_{\delta})$-approximation,
	where $\varepsilon'_{\delta} \approx \frac{6\delta (2-\delta)}{(1-\delta)^2}$.
\end{theorem}
\begin{proof}
	For any run of the algorithm, i.e., a fixed set $S$, the corresponding partial 
	realization $\omega_{|S}$ and the coin flips observed, define for convenience 
	the set $C$ as those items in $\{a_1,\dots,a_k\}$ that have been considered 
	during the algorithm and then not added to $S$ because of the coin flips. 
	Define $U= \{a_1,\dots,a_k\}\setminus C$.
	Additionally, define $C'$ to be the set of all items that are considered, but not 
	chosen during the run of our algorithm which have positive expected marginal 
	contribution to $S$. I.e., $C$ captures the items from the 
	\emph{good-looking} set after choosing $S$ that we missed due to coin 
	tosses, and $C'$ \emph{all} items we missed for the same reason which 
	should have had a positive contribution in hindsight. Note that
	$C \subseteq C'$.
	
	We can then split the left hand side term of \cref{two} into two parts: the sum 
	over $C$ (upper bounded by the sum over $C'$), and the sum over $U.$ Now we control separately these terms using linear 
	combinations of $v(S)$ and $v(i^*).$
	\begin{lemma}
		$\mathbb{E}[v(S)]  \geq p \cdot \mathbb{E}\left[\sum_{a\in 
		C}v\left(a|(S,\omega_{|S})\right)\right]$
	\end{lemma}
	
	\begin{proof}
		Since $C \subseteq C'$ and $C'$ contains all considered elements with 
		nonnegative expected contribution to $S$, it is sufficient to show 
		$\mathbb{E}[v(S)] \geq p \cdot \mathbb{E}\left[\sum_{a\in 
		C'}v\left(a|(S,\omega_{|S})\right)\right].$
		
		We proceed as in \cref{lem:S_vs_sum}. Let's consider for each $a \in A$ all 
		the events $\Ea$ capturing the story of a run of the algorithm up to the 
		point element $a$ is considered (all the history if it is never considered).
		
		Let $G_a$ be the marginal contribution of element $a$ to the solution set 
		$S$.
		If $\Ea$ corresponds to a story in which element $a$ is not considered, 
		then it does not contribute - neither in the left, nor in the right hand side of 
		the inequality we are trying to prove. Else, let $(S_a, \omega_a)$ be the 
		partial solution when it is indeed considered:
		\begin{align*}
		\mathbb{E}[G_a|\Ea] &= p \cdot v(a|(S_a,\omega_a))  \geq p \cdot 
		\mathbb{E}[v(a|(S,\omega_{|S}))|\Ea].
		\end{align*}
		The statement follows from the law of total probability with respect to 
		$\Ea$, and state-wise submodularity of $v$.
	\end{proof}

	\begin{lemma}
		$\mathbb{E}\left[ v(S) \right] +  v(i^*) \geq \mathbb{E}\left[ \sum_{a\in U} 
		v\left(a|(S,\omega_{|S}) \right) \right].$
	\end{lemma}
	
	\begin{proof}
		
		Now let us turn towards the items $U$ that were not considered by the 
		algorithm. The intuition behind the claim is that if they were not considered 
		then they were not good enough, in expectation, to compare with $S$. The 
		proof, though, has to deal with some probabilistic subtleties.
		
		Let's start fixing a story of the algorithm, i.e., the coin tosses and $(S, 
		\omega_S)$, $S=s_1, s_2, \dots, s_T$, numbered according to their 
		insertion in S, i.e., $s_i$ is the $i^{th}$ element to be added to $S$. 
		For the sake of simplicity let's also renumber the elements in $U$ as $a_1, 
		\dots, a_l$ respecting the order given by the marginals over costs. 
		
		There are two cases. If during the whole algorithm the elements in $U$ 
		have ratio $\frac{v\left(a|(S_i, \omega_{|S_i})\right)}{c_a}$ smaller than that of 
		the item which was instead considered, then one can easily argue, by 
		adaptive submodularity, that:
		\begin{align*}
		\sum_{a\in U} v\left(a|(S,\omega_{|S}) \right) &\leq \sum_{t=1}^T v(s_t|(S_t, 
		\omega_t)) + v(u_1|(S, \omega_{|S}) \leq \\
		&\leq \sum_{t=1}^T v(s_t|(S_t, \omega_t)) + v(u_1)  \leq \sum_{t=1}^T 
		v(s_t|(S_t, \omega_t)) + v(i^*).
		\end{align*}
		being $S_t=(s_1,\dots, s_{t-1})$ and $\omega_t$ the restriction of 
		$\omega_{|S}$ to $S_t.$
		Note that the last element $u_1$ is added to account for the unspent 
		budget by the solution. 
		We claim that the above inequality holds also in the case in which there is 
		an element in $U$ whose marginal over cost is greater than that of some in 
		$S$. Such an element can exist because of the budget constraint: during 
		the algorithm it had better marginal over cost, but was discarded because 
		there was not enough room for it. We observe there can exist at most one 
		such element, due to the budget constraint and because its value is upper 
		bounded by $u_1$, so the above formula still holds.
		
		Once we know that, by law of total probability, we have
		
		\begin{equation}\label{four}
		\mathbb{E}\left[ \sum_{a\in U} v\left(a|(S,\omega_{|S}) \right) \right] \leq 
		\mathbb{E}\left[v(S)\right] + v(i^*) ,
		\end{equation}
		concluding the proof.
	\end{proof}
	
	Combining the two Lemmata we get:
	\begin{align*}
	(1+\tfrac 1p) \,\mathbb{E}\left[ v(S) \right] + \mathbb{E}\left[ v(i^*) \right] &\geq  
	\mathbb{E}\left[ \sum_{a\in U} v\left(a|(S,\omega_{|S}) \right) 
	\right]+\mathbb{E}\left[ \sum_{a\in C} v\left(a|(S,\omega_{|S}) \right) \right]=\\
	&=\mathbb{E}\left[ \sum_{a\in U\cup C} v\left(a|(S,\omega_{|S}) \right) \right].
	\end{align*}
	Equation \cref{two} implies
	\[
	(1+\tfrac 1p) \,\mathbb{E}\left[ v(S) \right] + \mathbb{E}\left[ v(i^*) \right] \geq
	\mathbb{E}\left[ \sum_{a\in U\cup C} v\left(a|(S,\omega_{|S}) \right) \right]\geq 
	\mathbb{E}\left[ v\left(S^*|(S,\omega_{|S})\right)\right] .\]
	Also, by \cref{eq:adaptive_sampling} and some algebra:
	\begin{align*}
	\mathbb{E}\left[ v\left(S^*|(S,\omega_{|S})\right)\right]& = \mathbb{E}\left[ 
	\mathbb{E}\left[ v(S^*\cup S,\omega)-v(S,\omega)|\omega_{_S}  \right] \right]\\
	& =\mathbb{E}\left[ v(S^*\cup S) \right]-\mathbb{E}\left[ v(S) \right]\\
	&\geq (1 - p)  \cdot \mathbb{E}[v(S^*)]-\mathbb{E}\left[ v(S)\right]
	\end{align*}
	All together, denoting as $OPT$ the $\mathbb{E}[v(S^*)]$, we get:
	\begin{equation}
	\label{eq:adapt_largeinst}
	(2p+1)\mathbb{E}[v(S)]+p\mathbb{E}[v(i^*)] \geq p(1-p) OPT
	\end{equation}
	
	Let's call $ALG$ the expected value of the solution output by the algorithm. 
	Since the algorithm chooses with a coin flip either the best expected single 
	item or $S$, it holds
	\[
	ALG =  (1-p_0)v(S) + p_0 v(i^*)
	\]
	Picking $p_0 = \tfrac{p}{3p+1}$,
	\[
	ALG = \frac{2p+1}{3p+1}\mathbb{E}[v(S)]+\frac{p}{3p+1}\mathbb{E}[v(i^*)] \geq 
	\frac{p(1-p)}{3p+1}OPT.
	\]
	The right hand side is minimized for $p=\tfrac 13$, concluding the proof of the 
	first part of the statement.
	
	The lazy version of \textsc{AdaptiveGreedy} is analogous to the non-adaptive 
	setting, both for the algorithm and the analysis, so we omit repeating the 
	proof.

	In order to prove the last part of the statement, we start from 
	\cref{eq:adapt_largeinst} and apply the large instance property:
	$$
	(1-p)\E{v(S^*)} \leq \left(2+\tfrac 1p\right) \E{v(S)} + \E{v(i^*)} \leq \left(2+\tfrac 
	1p\right) \E{v(S)} + \delta \cdot \E{v(S^*)}
	$$
	Rearranging terms and assuming $p + \delta < 1$ we have that:
	$$
	\E{v(S^*)} \leq \frac{\left(2+\tfrac 1p\right)}{1-p-\delta} \cdot \E{v(S)}. 
	$$
	Optimizing for $p \in (0,1-\delta)$, we get the claimed result. Specifically, for 
	${p=(\sqrt{3-2\delta}}-1)/2$ the approximation factor is $(4+2\sqrt{3} + 
	\varepsilon_{\delta})$, with 
	$$
	\varepsilon_{\delta}= 
	2\left(\frac{\sqrt{3-2\delta}+1}{(1-\delta)^2}+\frac{1}{1-\delta}-\sqrt 3 - 2\right) 
	\approx \frac{6\delta (2-\delta)}{(1-\delta)^2}.
	$$
\end{proof}

\section{Experiments}
\label{sec:experiments}

Out of the numerous applications of submodular maximization subject to a 
knapsack constraint, we evaluate \textsc{SampleGreedy} and  
\textsc{AdaptiveGreedy} on two selected examples, using real and synthetic 
graph topologies. Variants of these have been studied in a similar context; see 
\citep{MirzasoleimanBK16}.

As our algorithms are randomized, but extremely fast, we use the best output 
out of $n=5$ iterations. A delicate point is tuning the probabilities of acceptance 
$p$ (\textcolor{links}{line} \ref{line:r_i-adaptive} of \textsc{AdaptiveGreedy}) for improved 
performance.
While the choices of $p$ in \cref{thm:greedy,thm:adaptive} minimize our analysis 
of the theoretical worst-case approximation, there are two factors that suggest a 
value much closer to 1 works best in practice: the small value of any singleton 
solution, and the much better guarantee of  
\cref{lem:Feige+} 
for most widely used non-monotone submodular objectives. 
We do not micro-optimize for $p$ but rather choose uniformly at randomly from 
$[0.9, 1]$.

\newcommand{\helpv}{f}
\paragraph{Video Recommendation:}
Suppose we have a large collection $A$ of videos from various categories 
(represented as possibly intersecting subsets $C_1, \ldots, C_k \subseteq A$) 
and we want to design a recommendation system. When a user inputs a subset 
of categories and a target total length $B$, the system should return a set of 
videos from the selected categories of total duration at most $B$ that maximizes 
an appropriate objective function. (Of course, instead of time here, we could use 
costs and a budget constraint.)
Each video has a rating and there is some measure of similarity between any 
two videos.  We use a weighted graph on $A$ to model the latter: each edge 
$\{i,j\}$ between two videos $i$ and $j$ has a weight $w_{ij}\in [0,1]$ capturing 
the percentage of their similarity.
To pave the way for our $v(\cdot)$, we start from the auxiliary objective 
$\helpv(S)= \sum_{i\in S}\sum_{j\in A} w_{ij} - \lambda \sum_{i\in S}\sum_{j\in S} 
w_{ij}$,
for some $\lambda\ge 1$ \citep{LinB10,MirzasoleimanBK16}. This is a 
\emph{maximal marginal relevance} inspired objective \citep{CarbonellG17} that 
rewards coverage, while penalizing similarity. 
For $\lambda = 1$, internal similarities are irrelevant and $\helpv$ becomes a cut 
function. However, one can penalize similarities even more  severely as $\helpv$ 
is submodular for $\lambda\ge 1$ (e.g., \citet{LinB10} use $\lambda = 5$).

In order to mimic the effect of a partition matroid constraint, i.e., the avoidance 
of many videos from the same category, we may use two parameters 
$\lambda\ge 1, \mu \ge 0$. While $\lambda$ is as above, $\mu$ puts extra 
weight on similarities between videos that belong to the same category. That 
leads to a  more general auxiliary objective 
$g(S)= \sum_{i\in S}\sum_{j\in A} w_{ij} -  \sum_{i\in S}\sum_{j\in S} (\lambda + 
\chi_{ij}\mu)w_{ij}$,
where $\chi_{ij}$ is equal to $1$ if there exists $\ell$ such that $i, j \in C_{\ell}$ 
and $0$ otherwise.
To interpolate between choosing highly rated videos and videos that represent 
well the whole collection, here we use the submodular function $v(S) = \alpha 
\sum_{i\in S} \rho_i + \beta g(S)$ for $\alpha, \beta \ge 0$, where $\rho_i$ is the 
rating of video $i$. We use $\lambda =  3$, $\mu = 7$ and set the parameters 
$\alpha, \beta$ so that the two terms are of comparable size.

We evaluate \textsc{SampleGreedy} on an instance based on the
latest version of the 
\href{https://grouplens.org/datasets/movielens/25m/}{MovieLens dataset} 
\citep{movielens}, which includes 62000 movies, 13816 of which have 
\emph{both} user-generated tags \emph{and} ratings.
We calculate the weights $w_{ij}$ using these tags (with the L2 norm of the 
pairwise minimum tag vector, see \cref{app:experiments}) while the costs 
are drawn independently from $U(0,1)$.
We compare against the FANTOM algorithm of \citet{MirzasoleimanBK16} as it 
is the only other algorithm with a provable approximation guarantee that runs in 
reasonable time. Continuous greedy approaches \citep{FeldmanNS11} or the 
repeated greedy of \citet{GuptaRST10} are prohibitively slow. 
\textsc{SampleGreedy} consistently performs  better than FANTOM for a wide 
range of budgets (\cref{fig:sfig1}). Plotting the number of function 
evaluations against the budget, \textsc{SampleGreedy} is much faster (\cref{fig:sfig4}) despite the fact that it is run $5$ times!

\begin{remark}\label{rem:lazyfantom}
	The running time of FANTOM for fixed $\varepsilon$ is $O(nr\log n)$, where 
	$r$ is the cardinality of the largest feasible solution. For a knapsack constraint 
	this translates to $O(n^2\log n)$. To be as fair as possible, we implemented 
	FANTOM using lazy evaluations, which improves the number of evaluations 
	of the objective function to $O(n\log^2 n)$ and is indeed much faster in 
	practice, for the knapsack sizes we consider. Even so, our 
	\textsc{SampleGreedy} is faster by a factor of $\Omega(\log n)$ which, 
	including the improvement in the constants involved, still makes a huge 
	difference. Note that in both \cref{fig:sfig4,fig:sfig5} one can 
	discern the superlinear increase of the function evaluations for FANTOM but 
	not for \textsc{SampleGreedy}.
\end{remark}

\setlength{\belowcaptionskip}{-10pt}

\begin{figure}[ht!]
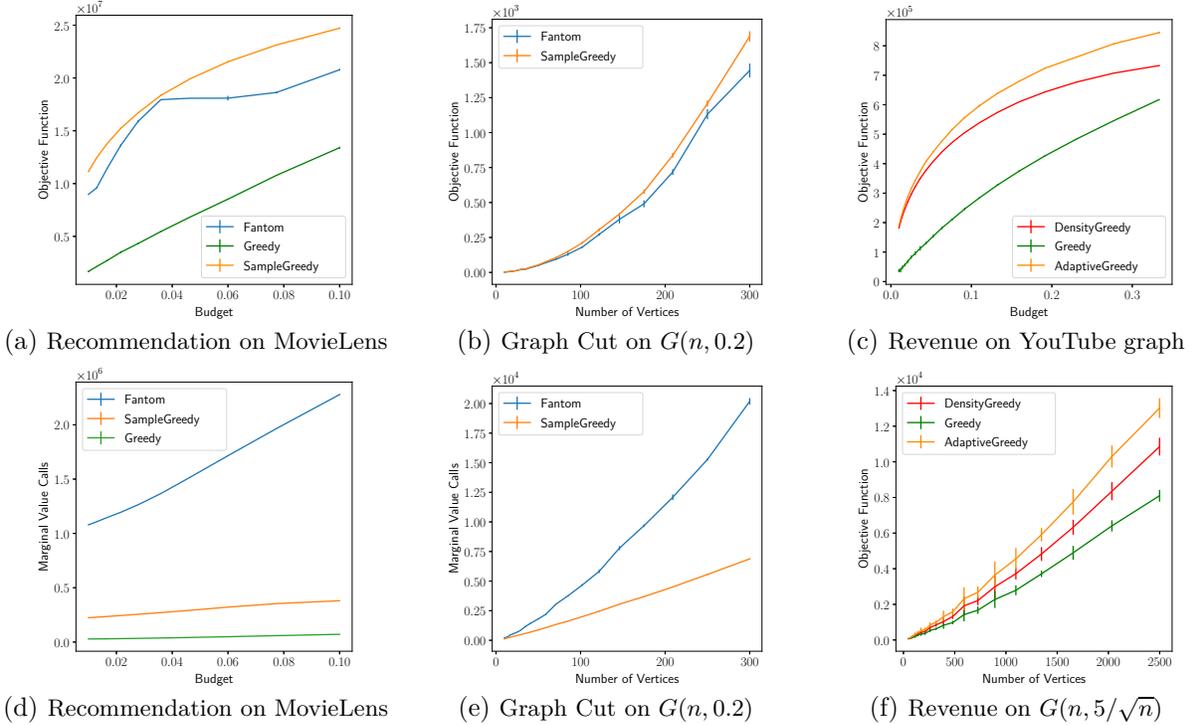

	\captionsetup[subfigure]{aboveskip=-3pt}
	\centering
	\begin{subfigure}{.33\textwidth}
		\centering
		\scalebox{0.33}{\input{movieLens_big_objective.pgf}}
		\caption{\footnotesize Recommendation on MovieLens}
		\label{fig:sfig1}
	\end{subfigure}\hspace{0.1pt}%
	\begin{subfigure}{.33\textwidth}
		\centering
		\scalebox{0.33}{\input{marginal_calls_gnp_objective.pgf}}
		\caption{\footnotesize Graph Cut on $G(n, 0.2)$}
		\label{fig:sfig2}
	\end{subfigure}\hspace{0.1pt}%
	\begin{subfigure}{.33\textwidth}
		\centering
		\scalebox{0.33}{\input{youtube_adaptive_objective.pgf}}
		\caption{\footnotesize Revenue on YouTube graph}
		\label{fig:sfig3}
	\end{subfigure}\\ \vspace{8pt}
	\begin{subfigure}{.33\textwidth}
		\centering
		\scalebox{0.33}{\input{movieLens_big_marginal.pgf}} 
		\caption{\footnotesize Recommendation on MovieLens}
		\label{fig:sfig4}
	\end{subfigure}\hspace{0.1pt}%
	\begin{subfigure}{.33\textwidth}
		\centering
		\scalebox{0.33}{\input{marginal_calls_gnp_marginal.pgf}} 
		\caption{\footnotesize Graph Cut on $G(n, 0.2)$}
		\label{fig:sfig5}
	\end{subfigure}\hspace{0.1pt}%
	\begin{subfigure}{.33\textwidth}
		\centering
		\scalebox{0.33}{\input{fixed_budget_adaptive_objective.pgf}}
		\caption{\footnotesize Revenue on $G(n, 5/\sqrt{n})$}
		\label{fig:sfig6}
	\end{subfigure}\vspace{8pt}%
	\caption{\small The four plots on the left compare the performance and the 
	number of function evaluations of \textsc{SampleGreedy} and FANTOM on 
	the video recommendation problem for the MovieLens dataset (a), (d) and on 
	the maximum weighted cut problem on random graphs (b), (e). Since no 
	$\varepsilon\le 1$ affected the performance of FANTOM noticeably before 
	becoming too computationally expensive, we used $\varepsilon = 1$ to 
	achieve the maximum possible speedup.  The plots on the far right illustrate 
	the performance of \textsc{AdaptiveGreedy} (ignoring single item solutions, 
	i.e., $p_0 = 0$) on the influence-and-exploit problem for two distinct 
	topologies: the YouTube graph (c) and random graphs (f). All budgets are 
	shown as fractions of the total cost.} 
\end{figure}
\paragraph{Influence-and-Exploit Marketing:}
Consider a  seller of a single digital good (i.e., producing extra units of the good 
comes at no extra cost) and a social network on a set $A$ of potential buyers. 
Suppose that the buyers influence each other and this is quantified by a weight 
$w_{ij}$ on each edge $\{i, j\}$ between buyers $i$ and $j$. 
Each buyer's value for the good depends on who owns it within her immediate 
social circle and how they influence her. 
A possible revenue-maximizing strategy for the seller is to first give the item for 
free to a selected set $S$ of influential buyers (influence phase) and then extract 
revenue by selling to each of the remaining buyers at a price matching their value 
for the item due to the influential nodes (exploit phase). 
Here we further assume, similarly to the adaptation of this model by 
\citet{MirzasoleimanBK16}, that each buyer comes with a cost of convincing her 
to advertise the product to her friends. The seller has a budget $B$ and the  set 
$S$ should be such that $\sum_{i\in S} c_i \le B$.

We adopt the generalization of the \emph{Concave Graph Model} of 
\citet{HartlineMS08} to non-monotone functions \citep{BabaeiMJS13}. Each 
buyer $i\in A$ is associated with a non-negative concave function $f_i$. For any 
$i\in A$ and any set $S\subset A\mysetminus\{i\}$ of agents already owning the 
good, the value of $i$ for it is $v_i(S) = f_i\big(\sum_{j\in S\cup\{i\}}w_{ij}\big)$. 
The total potential revenue $v(S) = \sum_{i\in A\mysetminus S}v_i(S)$ that we aim to maximize is a non-monotone submodular function.
Besides the theoretical guarantees for influence-and-exploit marketing in the 
Bayesian setting \citep{HartlineMS08}, there are strong experimental evidence of 
its  performance in practice \citep{BabaeiMJS13}.
The problem generalizes naturally to different stochastic versions. We assume 
that the valuation function of each buyer $i$ is of the form $f_i(x)=a_i \sqrt{x}$ 
where $a_i$ is drawn independently from a Pareto Type II distribution with $\lambda=1$, $\alpha=2$.  We only learn the exact value of a buyer when we give the good for free to someone in her neighborhood. 

We evaluate \textsc{AdaptiveGreedy} on an instance based on the 
\href{https://snap.stanford.edu/data/com-Youtube.html}{YouTube graph} 
\citep{yang2015defining}, containing 1,134,890 vertices. The (known) weights are 
drawn independently from $U(0,1),$ and the costs are proportional to the sum of 
the weights of the incident edges. As \textsc{AdaptiveGreedy} is the first 
adaptive algorithm for the problem, we compare with non-adaptive alternatives 
like  \emph{Greedy}\footnote{The simple greedy algorithm that in each step 
picks the element with the largest marginal value.} and \emph{Density 
Greedy}\footnote{The greedy part of Wolsey's algorithm \citep{Wolsey82}.} for 
different values of the budget.
\textsc{AdaptiveGreedy} outperforms  the alternatives by up to $20\%$ (\cref{fig:sfig3}).

We observe similar improvements for Erd\H{o}s-R\'{e}nyi random graphs of 
different sizes and edge probability $5/\sqrt{n}$ and a fixed budget of $10\%$ of the total cost (\cref{fig:sfig6}).

\paragraph{Maximum Weighted Cut:}
Beyond the above applications, we would like to compare 
\textsc{SampleGreedy} to FANTOM with respect to both their performance and 
the number of value oracle calls as $n$ grows. 
We turn to \emph{weighted cut functions}---one of the most prominent 
subclasses of non-monotone submodular functions---on dense Erdős–Rényi 
random graphs with edge probability $0.2$. The  weights and the costs are 
drawn independently and uniformly from $[0,1]$ and the budget is fixed to 
$15\%$ of the total cost. Again \textsc{SampleGreedy} consistently performs 
better than FANTOM, albeit by $5$--$15\%$ (\cref{fig:sfig2}). In terms of 
running time, there is a large difference in favor of \textsc{SampleGreedy} (even 
for multiple runs), while the superlinear increase for FANTOM is evident (
\cref{fig:sfig5}).

\section{Conclusions}
The proposed random greedy method yields a considerable improvement over 
state-of-the-art algorithms, especially, but not exclusively, regarding the 
handling of huge instances. With all the subtleties of our work affecting solely 
our analysis, the algorithm remains strikingly simple and we are confident this 
will also contribute to its use in practice.
Simultaneously, this very simplicity translates into a generality that can be 
employed to achieve comparably good results for a variety of settings; we  
demonstrated this in the case of the adaptive submodularity setting.

Specifically, we expect that our approach can be directly utilised to improve the 
performance and running time of algorithms that now use some variant of the 
algorithm of \citet{GuptaRST10}. Such examples include the distributed 
algorithm of \citet{BarbosaENW15} and the streaming algorithm of 
\citet{MirzasoleimanJK18} in the case of a knapsack constraint. 
We further suspect that the same algorithmic principle can be  applied in the 
presence of incentives. This would largely improve the current state of the art in 
budget-feasible mechanism design for non-monotone objectives 
\citep{ChenGL11,AmanatidisKS19}.

Finally, a major question here is whether the same high level approach is valid 
even in the presence of additional combinatorial constraints. In particular, is it 
possible to achieve similar guarantees as FANTOM for a $p$-system and 
multiple knapsack constraints using only $O(n\log n)$ value queries?

\pagebreak
\bibliography{submodular}
\pagebreak

\appendix

\section{Additional Details on \cref{sec:experiments}}
\label{app:experiments}
 All graphs contain error bars, indicating the standard deviation between 
different runs of the experiments. This is usually insignificant due to the 
concentrating effect of the large 
size of the instances, despite the randomly 
initialized weights and inherent randomness of the algorithms used. 
Nevertheless, all results are obtained by running each experiment a number of 
times. For all algorithms involved, we use  \emph{lazy evaluations}  with 
$\varepsilon = 0.01$. 

\paragraph{Video Recommendation:}
We expand on the exact definition of the similarity measure that is only tersely 
described in the main text. Each movie $i$ is associated with a tag vector 
$t^i \in [0,1]^{1128}$, where each coordinate contains a relevance score for that 
individual tag. These tag vectors are \emph{not} normalized and have no 
additional structure, other than each coordinate being restricted to $[0,1]$. We 
define the similarity $w_{ij}$ between two movies $i, j$ as:
\begin{equation*}
w_{ij} = \sqrt{\sum_{k=1}^{1128} \left(\min \{t^i_k, t^j_k \}\right)^2}.
\end{equation*}
In other words, it is the L2 norm of the \emph{coordinate-wise minimum} of $t^i$ 
and $t^j$. This metric was chosen so that if \emph{both} movies have a high 
value in some tag, this counts as a much stronger similarity than one having a 
high value and the other a low one. For example, if we consider an inner product 
metric, any movie with all tags set to 1 would be as similar as possible to all 
other movies, even though it would include many tags that would be missing 
from the others. In particular, any movie would be appear more similar to the all 
1 movie than to itself! Choosing the minimum of both tags avoids this issue. 
Another possibility would be to normalize each tag vector before taking the inner 
product, to obtain the \emph{cosine similarity}. Although this alleviates some of 
the issues, there is some information loss as one movie could meaningfully have 
higher scores in all tags than another one; tags are not mutually exclusive. 
Ultimately any sensible metric has advantages and disadvantages and the exact 
choice has little bearing on our results. The similarity scores are then divided by 
their maximum as a final normalization step.

The experiment was repeated 5 times. The budget is represented as a fraction of 
the total cost starting at 1/100 and geometrically increasing to 1/10 in 10 steps. 
The total computation time was around 3 hours.

\paragraph{Influence-and-Exploit Marketing:}
For the YouTube graph, the experiment was repeated 5 times for a budget 
starting 
at 1/100 of the total cost and geometrically increasing to 1/3 in 20 steps, leading 
to a total 
computation time of 7 hours. For the Erdős–Rényi 
graph with $n$ vertices and edge probability $5 / \sqrt{n}$ it was
repeated 10 times, for $n$ starting at 50 and geometrically increasing to 2500 in 
20 steps, taking approximately 10 minutes.

\paragraph{Maximum Weighted Cut:}
The experiment was repeated 10 times for $n$ starting at 10 and increasing 
geometrically to 300 in 20 steps, requiring approximately 5 minutes.

\end{document}